\documentclass[pdflatex,sn-mathphys-ay]{sn-jnl}
\usepackage{graphicx}
\usepackage{multirow}
\usepackage{amsmath,amssymb,amssymb, bm, amsfonts}
\usepackage{amsthm}
\usepackage{natbib}
\usepackage{mathrsfs}
\usepackage[title]{appendix}
\usepackage{xcolor}
\usepackage{textcomp}
\usepackage{manyfoot}
\usepackage{booktabs}
\usepackage{algorithm}
\usepackage{algorithmicx}
\usepackage{algpseudocode}
\usepackage{listings}%

\theoremstyle{thmstyleone}%
\newtheorem{theorem}{Theorem}

\theoremstyle{thmstyletwo}

\theoremstyle{thmstylethree}%

\begin{document}

\title[SMPtW]{A Novel Three-Parameter Extended Weibull Distribution for Health Data Modelling.}

\author*[1,4]{\fnm{Isqeel} \sur{Ogunsola}}\email{isqeel.ogunsola@manchester.ac.uk}
\author[2]{\fnm{Nurudeen} \sur{Ajadi}}\email{nurudeen.ajadi1@louisiana.edu}
\author[3]{\fnm{Gboyega} \sur{Adepoju}}\email{gadepoju@nova.edu}

\affil*[1]{\orgdiv{Department of Mathematics}, \orgname{University of Manchester}, {\country{United Kingdom}}}
\affil[2]{\orgdiv{Department of Mathematics}, \orgname{University of Louisiana at Lafayette, Louisiana}, {\country{USA}}}
\affil[3]{\orgdiv{Department of Mathematics}, \orgname{Nova Southeastern University, Davie, Florida}, {\country{USA}}}
\affil[4]{\orgdiv{Department of Statistics}, \orgname{Federal University of Agriculture, Abeokuta}, {\country{Nigeria}}}

\abstract{Weibull distribution is widely used in modelling health data. However, its lack of sufficient tail flexibility often results in poor fit in extreme events. We proposed another three-parameter extension of the Weibull distribution with additional flexibility without sacrificing tractability. We derived and studied its statistical properties, including reliability measures, quantile function, moment, stress-strength, mean waiting time, moment generating function, characteristics function, Rényi entropy, order statistics, mean residual life and mode. We adopted the inverse transform approach in random number generation, and through simulation, we evaluated the performance of the maximum likelihood estimates.  The fitness of the distribution was examined using a fracture dataset and compared with five similar extensions of the Weibull distribution. Our proposed novel distribution fits the data best among the competing models. It is therefore recommended as a better alternative in modelling heavily tailed health data due to its flexibility. Robust estimation techniques would be valuable in addressing potential numerical challenges associated with the model in future studies.}

\maketitle

\section{Introduction}
Researchers across disciplines, including medicine, health, biological, and environmental sciences, continually encounter complex real-world situations, and the main challenge is to identify their patterns. Probability distributions play a significant role in analysing real-world data across various areas of the applied sciences. These distributions are effective in analysing extreme situations and natural disasters such as earthquakes, weather, floods, and droughts (\cite{klakattawi2022survival}, \cite{chaito2021length}, and \cite{zare2021generating}). 

However, it is widely known that there is no single probability distribution that is able to provide a satisfactory fit in all situations \citep{almalki2013new}. Therefore, researchers have repeatedly considered proposing new methods to introduce new distributions with a good ability to fit various types of data in different applications by adding a new parameter to common probability distributions. Some of the old methods of achieving this include exponentiation \citep{gupta1998modeling}, transmutation \citep{shaw2007alchemy}, Marshall Olkin \citep{marshall1997new}, Gamma-generated families \citep{zografos2009families},  Alpha power transformation \citep{mahdavi2017new}, to mention a few. See \cite{lee2013methods} for a review of some of the old methods. Meanwhile, many new methods of distribution extensions have been proposed recently, which include \cite{rasool2025innovative}, Novel Exponentiated G-family distribution (NEGFD), \citep{mir2025modeling}, PNJ method \citep{ahad2024novel}, modified kies method \citep{al2020new}, logarithmic-U family method (\citet{zhao2023novel} and so on. While our study is mainly on SMP method of transformation \citep{rasool2025innovative}, we specifically highlighted some of these new methods introduced lately, which are yet to be explored in extending many base distributions. Note that PNJ and SMP methods are named after the authors using the first letters of their names. See \citep{ahad2024new,rasool2025innovative} for details about the abbreviation of the nomenclature. 

Weibull distribution is one of the useful distributions in reliability, engineering, medicine, and other areas in modelling time-to-event data. Its cumulative distribution function (CDF) and probability distribution function (PDF) are given respectively as:
\begin{equation*}
G(y) = \,\left\{ \begin{array}{l}
 1 - e^{ - \left( {\frac{y}{\beta }} \right)^\phi  } \,\,\,\,\,\,\,\,\,\,y \ge 0, \\ 
 0\,\,\,\,\,\,\,\,\,\,\,\,\,\,\,\,\,\,\,\,\,\,\,\,\,\,\,\,\,\,\,\,\,\,y < 0\,\,
 \end{array} \right.
\end{equation*}
and 
\begin{equation*}
g(y) = \left\{ \begin{array}{l}
 \frac{\phi }{\beta }\left( {\frac{y}{\beta }} \right)^{\phi  - 1} e^{ - \left( {\frac{y}{\beta}} \right)^\phi  } ,\,\,y \ge 0, \\ 
 0\,\,\,\,\,\,\,\,\,\,\,\,\,\,\,\,\,\,\,\,\,\,\,\,\,\,\,\,\,\,\,\,\,\,\,\,\,\,\,\,\,\,\,\,y < 0\,\, \\ 
 \end{array} \right.
\end{equation*}

Its relation to several other distributions gives it strength and wider applicability. However, the Weibull distribution failed to capture or model complex data structures and non-monotonic hazard rates. In particular, its lack of sufficient tail flexibility often results in poor fit in the extremes, which are of greatest practical importance. Our study introduces additional flexibility without sacrificing tractability, and this makes our proposed extension an improvement over the standard two-parameter Weibull distribution. 

A new method of extending the flexibility of a distribution was given by \cite{rasool2025innovative}. They stated that for a given random variable $Y$, the CDF and PDF of SMP transformed distribution are respectively given as: 

\begin{equation}
F(y;\lambda ) = \,\left\{ \begin{array}{l}
 \frac{{e^{{\left( {\log \lambda } \right)} {\bar G(y)}}  - \lambda }}{{1 - \lambda }}\,\,\,\,\,\,\,\,\,\, if\,\lambda  > 0,\,\,\lambda  \ne 1, \\ 
G(y)\,\,\,\,\,\,\,\,\,\,\,\,\,\,\,\,\,\,\,\,\,\,\,\,\,\,\,if\,\lambda  = 1\,\, \\ 
 \end{array} \right.
 \label{SMP CDF}
\end{equation}
and 

\begin{equation}
f(y;\lambda ) = \,\left\{ \begin{array}{l}
 \frac{{e^{{\left( {\log \lambda } \right)} {\bar G(y)}} \left( {\log \lambda } \right)g(y)}}{{\lambda  - 1}}\,\,\,\,\,\,\,\,\,\,if\,\lambda  > 0,\,\,\lambda  \ne 1, \\ 
g(y)\,\,\,\,\,\,\,\,\,\,\,\,\,\,\,\,\,\,\,\,\,\,\,\,\,\,\,\,\,\,\,\,\,\,\,\,\,\,\,\,\,\,\,\,if\,\lambda  = 1\,\, \\ 
 \end{array} \right.
 \label{SMP PDF}
\end{equation}
where $G(y)$ is the CDF of the base distribution, $ \bar G(y) $ is the survival function given as $1 - G(y)$ and $\lambda$ is the transforming parameter.

Following \cite{rasool2025innovative}, the PDF and CDF of the three-parameter Weibull distribution via SMP method can be obtained as:

\begin{equation}
F_{SMPtW}(y; \lambda, \phi, \beta ) = \,\left\{ \begin{array}{l}
 \frac{{e^{{\left( {\log \lambda } \right)} {  e^{-\left( \frac{y}{\beta }\right)^\phi} }}  - \lambda }}{{1 - \lambda }}\,\,\,\,\,\,\,\,\,\, if\,\lambda  > 0,\,\,\lambda  \ne 1, \\ 
1 - e^{-\left( \frac{y}{\beta }\right)^\phi}\,\,\,\,\,\,\,\,\,\,\,\,\,\,\,\,\,if\,\lambda  = 1\,\, \\ 
 \end{array} \right.
 \label{SMPtWcdf equation}
\end{equation}

\begin{equation}
f_{SMPtW}(y; \lambda, \phi, \beta ) = \,\left\{ \begin{array}{l}
  \frac{{e^{{\left( {\log \lambda } \right)} {  e^{-\left( \frac{y}{\beta }\right)^\phi} }}} }{{\lambda  - 1}} (\log \lambda) \frac{\phi }{\beta }\left( {\frac{y}{\beta }} \right)^{\phi  - 1} e^{ - \left( {\frac{y}{\beta }} \right)^\phi}\,\,\,\,\,\,\,\,\,\,if\,\lambda  > 0,\,\,\lambda  \ne 1, \\ 
\frac{\phi }{\beta }\left( {\frac{y}{\beta }} \right)^{\phi  - 1} e^{ - \left( {\frac{y}{\beta }} \right)^\phi  }\,\,\,\,\,\,\,\,\,\,\,\,\,\,\,\,\,\,\,\,\,\,\,\,\,\,\,\,\,\,\,\,\,\,\,\,\,\,\,\,\,\,\,\,\,\,\,\,\,\,\,\,\,\,if\,\lambda  = 1\,\, \\ 
 \end{array} \right.
 \label{SMPtWpdf equation}
\end{equation}

Our goals in this study are to (i) provide an updated review of SMP transformed distributions (ii) introduce a new extension of Weibull distributions through SMP model transformation (iii) derive and study the statistical properties extensively (iv) generate random numbers from the distribution using the inverse transform approach and study ML estimates via simulation and (v) show the power of this distribution over existing similar distributions in modeling health data. The PDF and CDF of SMPtW distribution are given in Figures \ref{fig1} and  \ref{fig2}, respectively.

\begin{figure}[H]
\centering
\begin{minipage}{0.48\textwidth}
    \centering
    \includegraphics[width=\linewidth]{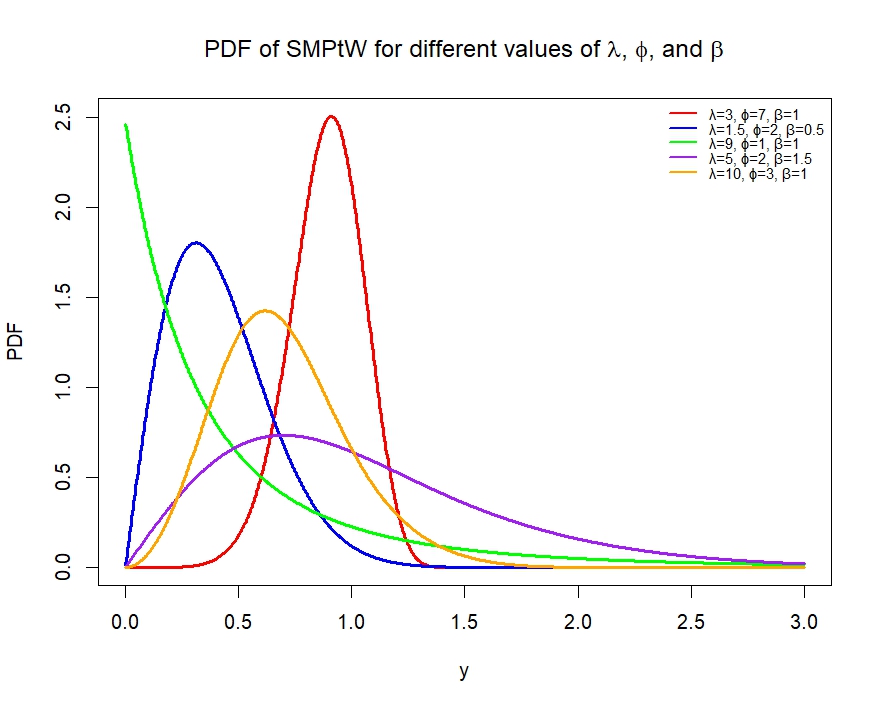}
    \caption{PDF of SMPtW}
    \label{fig1}
\end{minipage}\hfill
\begin{minipage}{0.48\textwidth}
    \centering
    \includegraphics[width=\linewidth]{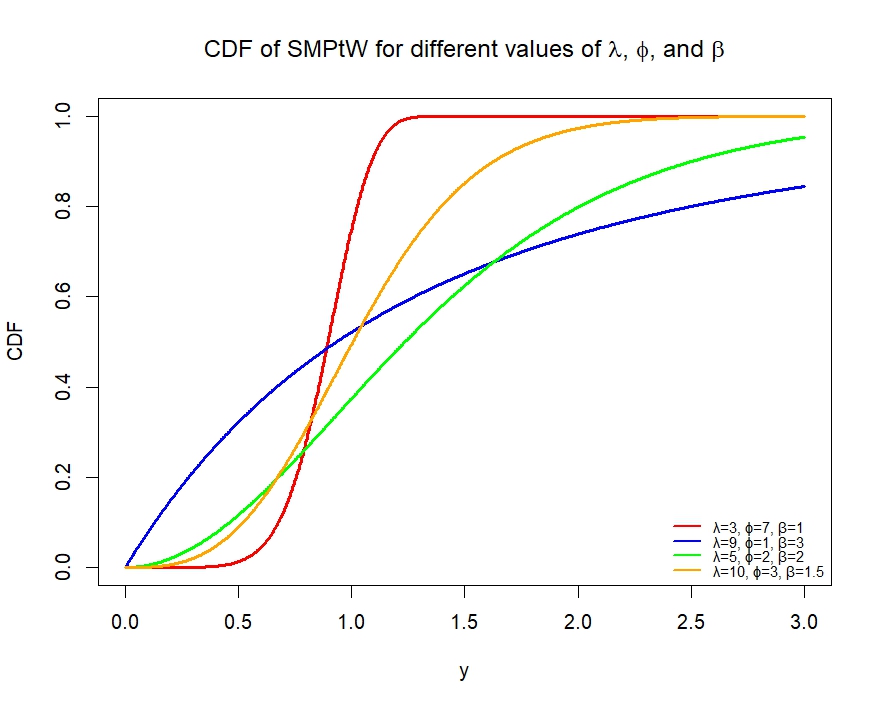}
    \caption{CDF of SMPtW}
    \label{fig2}
\end{minipage}
\end{figure}

The PDF and CDF plots for SMPtW distribution for different values of $\lambda$, $\phi $ and $\beta$ are shown in Figures \ref{fig1} and  \ref{fig2}. The PDF has right skew, exponential decay, unimodal, increasing and decreasing shapes, while the CDF has an increasing shape.

\vspace{2mm}
\noindent The proposed SMPtW distribution has the following sub-models.

\begin{itemize}

    \item  It becomes a Weibull distribution when $\lambda = 1, \phi > 0$ and $\beta > 0$.
    
    \item It reduces to an Exponential distribution when $\lambda = 1, \phi = 1$ and $\beta > 0 $.
    
    \item It becomes Rayleigh distribution if $\lambda = 1$, $\phi = 2$, and $\beta = \sqrt{2}\,\sigma$, where $\sigma > 0$ is the scale parameter of the Rayleigh distribution.
    
\end{itemize}

Other sections of the paper are structured as follows: Section 2 reviews some of the SMP transformed distributions previously studied and presents them in a fashion that is easy to compare at a glance. In Section 3, we introduced the novel three-parameter transformed Weibull distribution using the SMP approach named SMPtW distribution. Its important properties, such as reliability measures, quantile function, moments, stress-strength reliability, mean waiting time, moment-generating function, characteristic function, Rényi entropy, order statistics, mean residual life, and mode, were derived and studied extensively. In Sections 4 and 5, respectively, the MLE estimates were given, and a simulation study was carried out on the MLE using the bias and mean square error as evaluation metrics. In Section 6, the performance of the distribution was then examined using a health dataset and compared with other similar distributions. We discuss the result in Section 7. Finally, a concise conclusion and direction for future study were provided in Section 8.

\section{Review of SMP transformed distributions}
We provided an update on the review of SMP transformed distributions previously studied in this section. This is to keep the readers informed of the existing extended distributions through this approach and provide an opportunity to adopt this method in extending other base distributions yet to be explored via this approach. 

SMP method was introduced with Lomax exponential as an example \cite{rasool2025innovative}. They studied the properties of the new distributions in depth with application to two datasets and showed that the new extension has a better fit compared to Lomax and similar distributions. Ever since then, researchers have extended base distribution using this method. Other extended distributions through this method include: SMP Kumaraswamy \citep{jan2024new}, SMP Inverted Exponential \citep{ahad2024new},  SMP Pareto \citep{ahad2025new}. To the best of our knowledge, as at the moment of writing this paper, only four base distributions have been transformed through SMP. The PDFs of previously transformed SMP distributions are given in Table \ref{review}. The proposed distribution is then studied in Section 3.
\begin{table}[h]
\caption{Review of SMP transformed distributions}\label{review}%
\begin{tabular}{@{}lll@{}}
\toprule
SN & Distributions & Probability density function \\
\midrule
1 & SMP Lomax &  
$
f(z) =
\begin{cases}
\dfrac{e^{(log \gamma)(1+\frac{z}{\theta})^{-\phi}} (log \gamma) \frac{\phi}{\theta} (1+ \frac{z}{\theta})^{-(\phi+1)}}{\gamma-1}, & \gamma>0, \gamma\neq 1, \\
\frac{\phi}{\theta} (1+ \frac{z}{\theta})^{-(\phi+1)}, &  \gamma=1 \\
\end{cases}$ \\
2 & SMP Pareto   &  $
f(z) =
\begin{cases}
\dfrac{log (\gamma)}{\gamma-1} . \frac{\phi}{z^{\phi+1}}. e^{\frac{log (\gamma)}{z^\phi}}, & \gamma>0, \gamma\neq 1, \\
\frac{\phi}{z^{\phi+1}}, &  \gamma=1 \\
\end{cases}$ \\
3 & SMP Kumaraswamy         & $
f(z) =
\begin{cases}
\frac{\lambda \phi \text{log}  \gamma z^{\lambda - 1} (1-z^\lambda)^{\phi -1}e^{log\gamma (1-z^\gamma)^\phi}}{\gamma-1}, & \gamma>0, \gamma\neq 1, \\
\lambda \phi z^{\lambda - 1} (1-z^\lambda)^{\phi -1} &  \gamma=1 \\
\end{cases}$\\
4 & SMP Inverted Exponential  & $
f(z) =
\begin{cases}
  \frac{log (\gamma)}{\gamma - 1} \frac{\gamma}{z^2} e^{(log \gamma) (1-e^{-\gamma/z})} e^{-\gamma / z} & \gamma>0, \gamma\neq 1, \\
 \frac{\gamma}{z^2} e^{-\gamma / z}  &  \gamma=1 \\
\end{cases}$      \\
\botrule
\end{tabular}
\end{table}
\section{Properties of SMPtW Distribution}
The properties of SMPtW distribution are investigated in this section. We established the survival function, hazard function, quantile function, moment, moment generating function, characteristic function, mode, mean waiting time, mean residual life, stress-strength reliability, order statistics, and the Rényi entropy of the proposed SMPtW distribution.
\subsection{Reliability Measures: Survival and Hazard functions of SMPtW Distribution}

Survival function is defined as $ S(y)=1-F(y; \lambda, \phi, \beta )$. Using the CDF of SMPtW distribution in equation \ref{SMPtWcdf equation}, we obtained the survival function of SMPtW distributions as follows:

For $\lambda\neq 1$,
\[
S_{SMPtW}(y)
=1- \frac{{e^{{\left( {\log \lambda } \right)} {  e^{-\left( \frac{y}{\beta }\right)^\phi} }}  - \lambda }}{1- \lambda}
\]
\[
=\frac{{e^{{\left( {\log \lambda } \right)} {  e^{-\left( \frac{y}{\beta }\right)^\phi} }}  - 1 }}{\lambda - 1}
\]
For $\lambda=1$,
\[
S_{SMPtW}(y)=e^{-\left(\frac{y}{\beta}\right)^\phi}.
\]
Therefore,
\begin{equation}
S_{SMPtW}(y)=
\begin{cases}
\displaystyle \frac{{e^{{\left( {\log \lambda } \right)} {  e^{-\left( \frac{y}{\beta }\right)^\phi} }}  - 1 }}{\lambda - 1}
& \lambda>0,\ \lambda\neq 1,\\[1.2em]
\displaystyle e^{-\left(\frac{y}{\beta}\right)^\phi},
& \lambda=1.
\end{cases}
\label{survival equation}
\end{equation}

Also, The hazard function is defined as $$h_{SMPtW}(y)=\frac{f_{SMPtW}(y; \lambda, \phi, \beta )}{S_{SMPtW}(y)}.$$ Using equations \ref{SMPtWpdf equation} and \ref{survival equation}, the hazard function is also obtained as: 

For $\lambda\neq 1$,
\[
h_{SMPtW}(y)
=
\frac{
e^{{\left( {\log \lambda } \right)} {  e^{-\left( \frac{y}{\beta }\right)^\phi}}}
(\log\lambda)\frac{\phi}{\beta}\left(\frac{y}{\beta}\right)^{\phi-1}
e^{-\left(\frac{y}{\beta}\right)^\phi}
}{
{e^{{\left( {\log \lambda } \right)} {  e^{-\left( \frac{y}{\beta }\right)^\phi}}}}-1
}.
\]

For $\lambda=1$,
\[
h_{SMPtW}(y)=\frac{\phi}{\beta}\left(\frac{y}{\beta}\right)^{\phi-1}.
\]

Hence, the hazard function is given as:
\[
h_{SMPtW}(y)=
\begin{cases}
\frac{
e^{{\left( {\log \lambda } \right)} {  e^{-\left( \frac{y}{\beta }\right)^\phi}}}
(\log\lambda)\frac{\phi}{\beta}\left(\frac{y}{\beta}\right)^{\phi-1}
e^{-\left(\frac{y}{\beta}\right)^\phi}
}{
{e^{{\left( {\log \lambda } \right)} {  e^{-\left( \frac{y}{\beta }\right)^\phi}}}}-1
}
& \lambda>0,\ \lambda\neq 1,\\[1.4em]
\displaystyle
\frac{\phi}{\beta}\left(\frac{y}{\beta}\right)^{\phi-1} {  e^{-\left( \frac{y}{\beta }\right)^\phi}},
& \lambda=1.
\end{cases}
\]

\begin{figure}[H]
\centering
\begin{minipage}{0.48\textwidth}
    \centering
    \includegraphics[width=\textwidth]{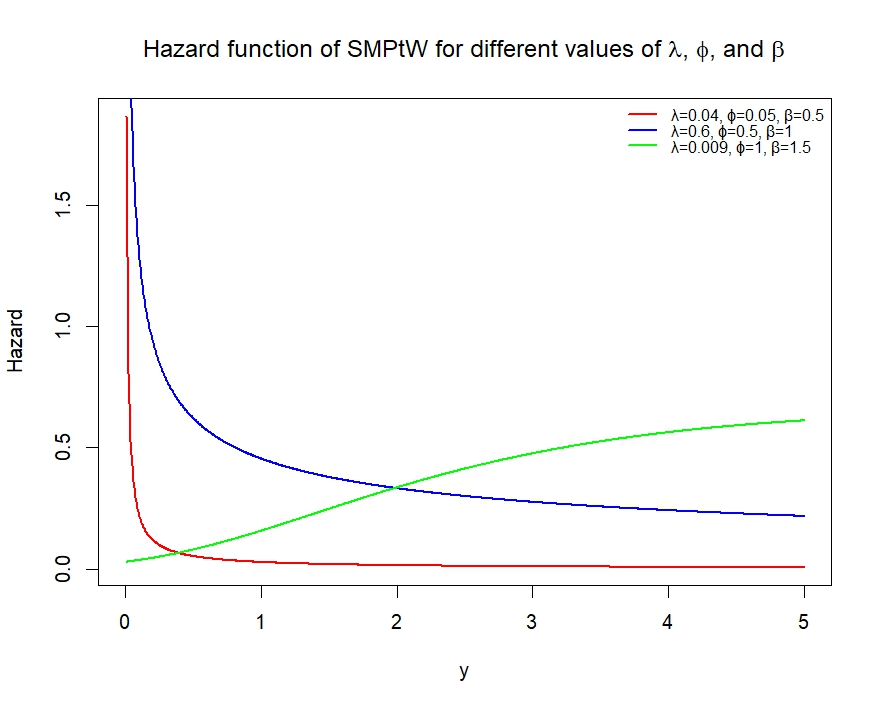}
\caption{Hazard function of SMPtW}\label{fig3}
\end{minipage}\hfill
\begin{minipage}{0.48\textwidth}
    \centering
  \includegraphics[width=\textwidth]{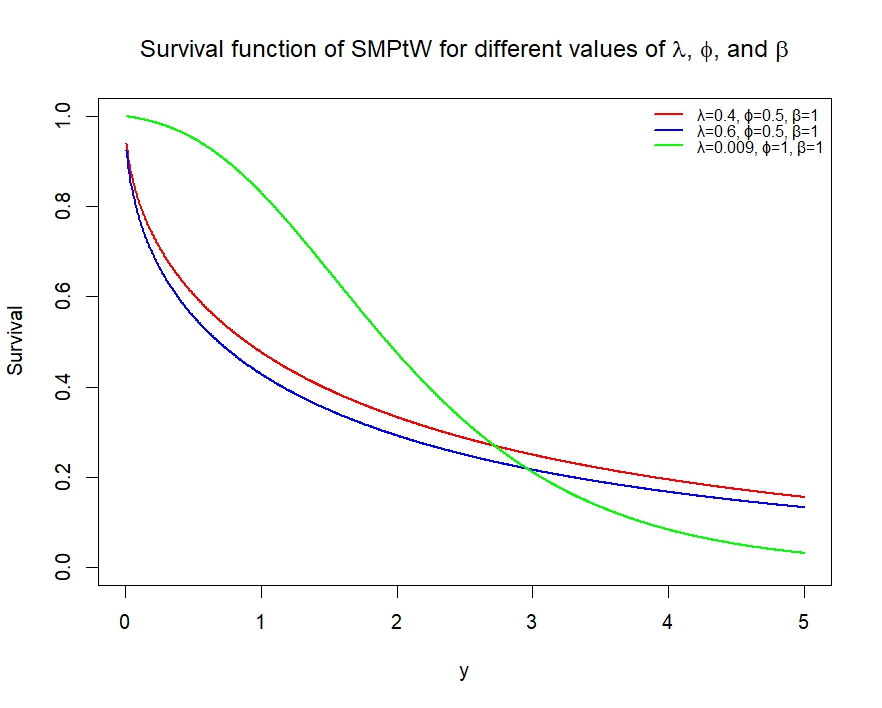}
\caption{Survival function of SMPtW}
    \label{fig4}
\end{minipage}
\end{figure}
The plots of survival and hazard functions of SMPtW distribution are shown in Figures \ref{fig3} and \ref{fig4}, respectively. The hazard function has increasing, decreasing, constant and non-monotonic shapes while the survival function has a decreasing shape.

\subsection{Quantile Function} 

\begin{theorem}
If $Y\sim SMPtW(\lambda,\phi, \beta)$ distribution, then the quantile function of $Y$ is given as

$$Q(u)=\beta \bigg[log \bigg(\dfrac{log\lambda}{log\{u(1-\lambda)+\lambda\}}\bigg)\bigg]^{1/\phi}, \hspace{5mm} 0<u<1$$
\end{theorem}

\begin{proof}
Set $u=F_{SMPtW}(y; \lambda, \phi, \beta )$, the quantile function can be obtained as follows:
$$
u=\dfrac{{e^{{\left( {\log \lambda } \right)} {  e^{-\left( \frac{y}{\beta }\right)^\phi} }}  - \lambda }}{1-\lambda} 
$$
Implies that $$ u(1-\lambda)+\lambda= e^{{\left( {\log \lambda } \right)} {  e^{-\left( \frac{y}{\beta }\right)^\phi} }} $$
Taking the logarithm of both sides, we have 

\[
(\log \lambda)\, e^{-\left(\frac{y}{\beta}\right)^\phi}
= \log\!\left(\lambda + u(1 - \lambda)\right).
\]
Simplifying, 
\[
e^{-\left(\frac{y}{\beta}\right)^\phi}
= \frac{\log\!\left(\lambda + u(1 - \lambda)\right)}{\log \lambda}.
\]

\[
-\left(\frac{y}{\beta}\right)^\phi
= \log\!\left(
\frac{\log\!\left(\lambda + u(1 - \lambda)\right)}{\log \lambda}
\right).
\]
We obtained the quantile function as:

\begin{equation}
    Q(u) = \beta \left[
-\log\!\left(
\frac{\log\!\left(\lambda + u(1 - \lambda)\right)}{\log \lambda}
\right)
\right]^{1/\phi}, \quad 0<u<1.
\label{quantile}
\end{equation}

Where $u$ follows a uniform $(0,1)$ distribution. 
\end{proof}
Hence, the median and other quartiles can be obtained using equation \ref{quantile}. For example, the median can be obtained by setting $u = 0.5$ as follows: 
\begin{equation*}
    Q(0.5) = \beta \left[
-\log\!\left(
\frac{\log\!\left(\lambda + 0.5(1 - \lambda)\right)}{\log \lambda}
\right)
\right]^{1/\phi}, \quad 0<u<1.
\end{equation*}

\subsection{Moment}
The moments provide a complete characterisation of the SMPtW distributions' shape and central tendencies.

\begin{theorem}
Suppose $Y \sim SMPtW(\lambda, \phi, \beta),$ then, the $r^{th}$ moment is given by

\begin{equation*}
E(Y^r)=\dfrac{log \lambda}{\lambda -1}\beta^r. \Gamma \left(\dfrac{r}{\phi} + 1\right)\sum_{j=0}^\infty \dfrac{(log \lambda)^j}{j !(j+1)^{r/\phi +1}} 
\end{equation*}
\end{theorem}
\begin{proof}
The $r^{th}$ moment of a random variable Y that follows SMPtW distribution is defined by \\
\[
\mu_r = E(Y^r)=\int_0^ \infty y^r .f_{SMPtW}(y, \lambda, \phi, \beta) dy
\]
\[
E(Y^r)=\int_0^ \infty y^r . \dfrac{e^{(log \lambda)e^{-(\frac{y}{\beta})^{\phi}}}(log \lambda) \frac{\phi}{\beta} (\frac{y}{\beta})^{\phi -1 }e^{-(\frac{y}{\beta})^{\phi}}}{\lambda-1} dy
\]

\[
=\dfrac{\phi log \lambda}{\beta(\lambda -1)} \int_0^ \infty y^r . (\frac{y}{\beta})^{\phi -1 } e^{-(\frac{y}{\beta})^{\phi}}e^{(log \lambda)e^{-(\frac{y}{\beta})^{\phi}}} dy \\
\]

Let $(\frac{y}{\beta})^{\phi-1}=\frac{y^{\phi-1}}{\beta^{\phi-1}}$ and \quad $y^r .(\frac{y}{\beta})^{\phi -1 } = y^{r+\phi-1} \beta^{-(\phi-1)}$\\
\[
E(Y^r)=\dfrac{\phi log \lambda}{\beta^\phi(\lambda -1)}\int_0^ \infty y^{r+\phi-1}e^{-(\frac{y}{\beta})^{\phi}}e^{(log \lambda)e^{-(\frac{y}{\beta})^{\phi}}} dy
\]
Let $t=(\frac{y}{\beta})^{\phi}$, so we have that $y=\beta t^{1/\phi}$, $dy=\dfrac{\beta}{\phi}t^{\frac{1}{\phi}-1} dt $.\\

$y^{r+\phi -1}=\beta ^{r+\phi -1} t^{\frac{r + \phi-1}{\phi}}$ ; $e^{-(\frac{y}{\beta})^\phi}=e^{-t}$ and $e^{(log \lambda)e^{-(\frac{y}{\beta})^{\phi}}}=e^{(log \lambda)e^{-t}}$

\[
E(Y^r)=\dfrac{\phi log \lambda}{\beta^\phi(\lambda -1)}\int_0^ \infty \beta^{r+\phi - 1} t^{\frac{r+\phi-1}{\phi}}e^{-t}e^{(log \lambda)e^{-t}} \frac{\beta}{\phi}t^{\frac{1}{\phi} -1} dt
\]

\[
=\dfrac{log \lambda}{\lambda -1}.\beta^r \int_0^ \infty t^{r/\phi}e^{-t}.e^{(log \lambda )e^{-t} } dt
\]
Using the power series,

\[
e^{(log \lambda)e^{-t}}=\sum_{j=0}^\infty \dfrac{(log \lambda)^j}{j !} e^{-jt}
\]
\[
E(Y^r)=\dfrac{log \lambda}{\lambda -1}.\beta^r \sum_{j=0}^\infty \dfrac{(log \lambda)^j}{j !} \int_0^ \infty t^{r/\phi}e^{-(j+1)t} dt
\]

Let $v=(j+1)t$, $dt = \dfrac{dv}{(j+1)}$, $t^{r/\phi}=\dfrac{v^{r/\phi}}{(j+1)^{r/\phi}}$.

\[
E(Y^r)=\dfrac{log \lambda}{\lambda -1} .\beta^r\sum_{j=0}^\infty \dfrac{(log \lambda)^j}{j !} \int_0^ \infty \dfrac{v^{r/\phi}.e^{-v}}{(j+1)^{r/\phi +1}} dv
\]

\[
=\dfrac{log \lambda}{\lambda -1} \beta^r\sum_{j=0}^\infty \dfrac{(log \lambda)^j}{j !(j+1)^{r/\phi +1}} \int_0^ \infty v^{r/\phi}.e^{-v} dv
\]
\begin{equation}
E(Y^r)=\dfrac{log \lambda}{\lambda -1}\beta^r. \Gamma \left(\dfrac{r}{\phi} + 1\right)\sum_{j=0}^\infty \dfrac{(log \lambda)^j}{j !(j+1)^{r/\phi +1}} 
\label{rthmoment}
\end{equation}
\end{proof}

Consequently, the first, second, third and fourth moments can be obtained using equation \ref{rthmoment} by setting $r = 1, 2,3 $ and $4$ respectively.

\subsection{Moment Generating Function}
The moment generating function (MGF) is used in obtaining the moments of a distribution, such as the mean and so on. We provided the MGF of the SMPtW distribution in this subsection.

\begin{theorem}
The MGF of a random variable say $Y$ that follows SMPtW distribution is denoted as $M_y(t)$ is
\begin{equation}
M_y(t)= \sum_{r=0}^{\infty} \frac{t^r}{r!} \dfrac{\log \lambda}{\lambda -1}\sum_{j=0}^\infty \dfrac{(log \lambda)^j}{j !(j+1)^{r/\phi +1}}  \Gamma \left(\dfrac{r}{\phi} + 1\right)
\end{equation}
\end{theorem}
\begin{proof}
The MGF of SMPtW distribution is given as:

\begin{equation}
    M_y(t) =  \int_0^\infty e^{ty}f_{\text{SMPtW}}(y;\lambda, \phi, \beta)dy 
\end{equation}
Expressing $e^{ty}$ in series form, we have; 
\begin{equation*}
 M_y(t) = \sum_{r=0}^{\infty} \frac{t^r}{r!}\int_0^\infty y^{r}f_{\text{SMPtW}}(y;\lambda, \phi, \beta)dy 
 \end{equation*}
Following the rth moment earlier obtained in equation \ref{rthmoment}, we obtained the MGF of SMPtW distribution as:
\begin{equation}
    M_y(t) = \sum_{r=0}^{\infty} \beta^r \frac{t^r}{r!} \dfrac{log \lambda}{\lambda -1} \Gamma \left(\dfrac{r}{\phi} + 1\right) \sum_{j=0}^\infty \dfrac{(log \lambda)^j}{j !(j+1)^{r/\phi +1}} 
\end{equation}
\end{proof}

\subsection{Characteristics Function}
The advantage of this property is that it exists for all distributions, even when the moment generating function does not exist. The characteristic function for a random variable $Y$ denoted as $ J_Y (t)$ is defined as:
\begin{equation*}
    J_Y (t) = E(e^{ity})
\end{equation*}
\begin{theorem}
Suppose $Y \sim SMPtW(\lambda, \phi, \beta) $, the characteristics function $J_Y (t)$ of SMPtW is given as
\begin{equation*}
    J_Y (t) = \dfrac{log \lambda}{\lambda -1} \sum_{r=0}^{\infty} \beta^r \frac{(it)^r}{r!} \Gamma \left(\dfrac{r}{\phi} + 1\right)  \sum_{j=0}^\infty \dfrac{(log \lambda)^j}{j !(j+1)^{r/\phi +1}}
\end{equation*}
\end{theorem}

\begin{proof}

By definition, the characteristic function of SMtW distribution is given as: 
\begin{equation*}
     J_Y (t) = \int_0^{\infty} e^{ity} f_{SMPtW}(y;\lambda, \phi, \beta)dy
\end{equation*}
By series expansion, we have;
\begin{equation*}
    J_Y (t) = \sum_{r=0}^{\infty} \frac{(it)^r}{r!}\int_0^{\infty} y^{r} f_{SMPtW}(y;\lambda, \phi, \beta)dy
\end{equation*}
Now, using equation \ref{rthmoment} we obtain;
\begin{equation}
    J_Y (t) = \dfrac{log \lambda}{\lambda -1} \sum_{r=0}^{\infty} \beta^r \frac{(it)^r}{r!} \Gamma \left(\dfrac{r}{\phi} + 1\right)  \sum_{j=0}^\infty \dfrac{(log \lambda)^j}{j !(j+1)^{r/\phi +1}}
\end{equation}
\end{proof}

\subsection{Mode of SMPtW Distribution}
Given a random variable $Y$ which follows the SMPtW distribution, the mode of the distribution is the value of $y$ for which $\frac{d}{dx}f(y) = 0$ and $\frac{d^2}{dx^2}f(y) < 0.$ However, given the positive nature of the PDF, the mode can also be obtained by finding the value of $y$ for which the log of the PDF derivative equals zero, i.e. $ \frac{d}{dy} logf(y) =0$. Using equation \ref{SMPtWpdf equation}, we derived the mode as follows:

\begin{align*}
    Log ( f_{SMPtW}(y; \lambda, \phi, \beta ) )& = log \left(\frac{{e^{{\left( {\log \lambda } \right)} {  e^{-\left( \frac{y}{\beta }\right)^\phi} }}} }{{\lambda  - 1}} (\log \lambda) \frac{\phi }{\beta }\left( {\frac{y}{\beta }} \right)^{\phi  - 1} e^{ - \left( {\frac{y}{\beta }} \right)^\phi} \right) \\
    & = (log \lambda)e^{-\left( \frac{y}{\beta }\right)^\phi}- (y/\beta)^{\phi}+ log(log \lambda) + log(\phi/\beta) \\ 
    &+  (\phi -1) log(y/\beta) - log (  \lambda-1)
\end{align*}
Now, differentiating with respect to y, we have;

\begin{align*}
   \frac{d Log f_{SMPtW}(y; \lambda, \phi, \beta )}{dy} &  = -(log \lambda) \frac{\phi}{\beta} \left(\frac{y}{\beta}\right)^{\phi - 1} e^{-\left(\frac{y}{\beta}\right)^{\phi}}- \frac{\phi}{\beta}\left(\frac{y}{\beta}\right)^{\phi - 1} + (\phi -1)/y 
    \end{align*}
Equating to zero, we have;

\begin{align*}
    0 & = -(log \lambda) \frac{\phi}{\beta} \left(\frac{y}{\beta}\right)^{\phi - 1} e^{-\left(\frac{y}{\beta}\right)^{\phi}}- \frac{\phi}{\beta}\left(\frac{y}{\beta}\right)^{\phi - 1} + (\phi -1)/y \\
    \phi - 1 & = \frac{\phi y}{\beta} \left(\frac{y}{\beta}\right)^{\phi - 1} \left[ (\log \lambda) e^{-\left(\frac{y}{\beta}\right)^{\phi}} +1 \right] \\
    & = \phi \left(\frac{y}{\beta}\right)^{\phi} \left[ (\log \lambda) e^{-\left(\frac{y}{\beta}\right)^{\phi}} +1 \right] 
 \end{align*}
 Therefore,
\begin{equation}
     y = \beta \left[ \frac{\phi - 1}{\phi[ (log \lambda) e^{- \left(\frac{y}{\beta}\right)^\phi} + 1]} \right]^{1/\phi}
     \label{mode}
\end{equation}

\textit{\textbf{Remark 2:}} The equation \ref{mode} is not in a closed form and can be obtained numerically. It is to be noted that if $\phi >1$, the density admits an interior mode given by the unique solution of \ref{mode}. If $0<\phi \leq 1$, and for $y>0$, all the derivatives will be negative with a decreasing density and mode at $y= 0$.

\subsection{Mean Waiting Time}
\begin{theorem}
The mean waiting time for a random variable $Y$ which follows the SMPtW distribution is given as:
\[
\bar\mu(t) = t+ \frac{\beta(log \lambda)}{e^{(log \lambda)e^-{\left(\frac{t}{\beta}\right)}^\phi}-\lambda} \sum_{j=0} ^\infty \frac{(log \lambda)^j}{j!(j+1)^{1+1/\phi}}\gamma \left( 1+\frac{1}{\phi}, (j+1)(t/\beta)^{\phi}\right)
\]
\end{theorem}

\begin{proof}
The mean waiting time is defined by 
\[
\bar\mu(t) = t- \frac{1}{F(t)} \int_0 ^t y.f_{SMPtW}(y;\lambda, \phi, \beta) dy
=t-\frac{1-\lambda}{e^{(log \lambda)e^{(\frac{-y}{\beta})^\phi}}-\lambda} \int_0 ^t y . \dfrac{e^{(log \lambda)e^{-(\frac{y}{\beta})^{\phi}}}(log \lambda) \frac{\phi}{\beta} (\frac{y}{\beta})^{\phi -1 }e^{-(\frac{y}{\beta})^{\phi}}}{\lambda-1} dy
\]
Let $u=-\left(\frac{y}{\beta}\right)^\phi$ ; $y=\beta u^{1/\phi}$ ; $dy=\frac{\beta}{\phi}u^{1/\phi-1} du$. In addition, we have the conditions that: If $ y=0$ then $u=0$ and if $ y=t$ then $u=(\frac{-t}{\beta})$\\

Using the power series
\[
e^{(log \lambda)e^{-u}}=\sum_{j=0} ^\infty \frac{(log \lambda)^j}{j!}e^{-ju}
\]
\[
\bar\mu(t) = t+\frac{\lambda -1}{\left(e^{(log \lambda)e^-\left(\frac{t}{\beta}\right)}-\lambda\right) \left(\lambda-1\right)} \beta(log \lambda) \sum_{j=0} ^\infty \frac{(log \lambda)^j}{j!} \int_0 ^u u^{1/\phi}.e^{-(j+1)u} du
\]
$v=(j+1)u$, $du=\frac{dv}{(j+1)}$.

The integral function can be expressed as 

\[
\int_0 ^u u^{1/\phi}.e^{-(j+1)u}du=\frac{1}{(j+1)^{1+1/\phi}}\gamma \left( 1+\frac{1}{\phi}, (j+1)(t/\beta)^{\phi}\right)
\]

\begin{equation}
\bar\mu(t) = t+ \frac{\beta(log \lambda)}{e^{(log \lambda)e^-{\left(\frac{t}{\beta}\right)}^\phi}-\lambda} \sum_{j=0} ^\infty \frac{(log \lambda)^j}{j!(j+1)^{1+1/\phi}}\gamma \left( 1+\frac{1}{\phi}, (j+1)(t/\beta)^{\phi}\right)
\end{equation}
\end{proof}

\subsection{Mean Residual Life of SMPtW Distribution}
\begin{theorem}
    The mean residual of the SMPtW distribution is given as \[
\mu(t)=\dfrac{-\beta log \lambda}{1-e^{(log \lambda)e^-(t/\beta)^\phi}} \sum_{j=0} ^\infty\dfrac{(log \lambda)^j}{j !(j+1)^{(1/\phi +1)}}\Gamma \left( 1+\frac{1}{\phi}, (j+1)(t/\beta)^{\phi}\right)-t
\]

\end{theorem}

\begin{proof}
The mean residual life of a  distribution can be defined as 
\[
\mu(t)=\frac{1}{S(t)}.\left( E(Y) - \int_0 ^t y.f_{SMPtW}(y; \lambda, \phi, \beta )\right)-t
\]
We defined the survival function $S(t)=1-F(t)$ below
\[
S(t) = \frac{1-e^{(log \lambda)e^{-(y/\beta)^\phi}}}{1-\lambda}
\]
From equation \ref{rthmoment}, The first moment is  \[E(Y)=\dfrac{\beta log \lambda}{\lambda -1} \Gamma \left(\dfrac{1}{\phi} + 1\right)\sum_{j=0}^\infty \dfrac{(log \lambda)^j}{j !(j+1)^{1/\phi +1}}{j!}\]
and \[
\int_0 ^t y.f_{SMPtW}(y,\lambda, \phi, \beta) dy=\dfrac{\beta log \lambda}{\lambda -1} \sum_{j=0} ^\infty\dfrac{(log \lambda)^j}{j !(j+1)^{1/\phi +1}}\gamma \left( 1+\frac{1}{\phi}, (j+1)(t/\beta)^{\phi}\right)
\]
By simplification, we have that:

\begin{equation}
\mu(t)=\dfrac{-\beta log \lambda}{1-e^{(log \lambda)e^-(t/\beta)^\phi}} \sum_{j=0} ^\infty\dfrac{(log \lambda)^j}{j !(j+1)^{(1/\phi +1)}}\Gamma \left( 1+\frac{1}{\phi}, (j+1)(t/\beta)^{\phi}\right)-t
\end{equation}

\end{proof}

\subsection{Stress-Strength Reliability}

Suppose that $Y_1$ and $Y_2$ are continuous random variables where $Y_1 \sim SMPtW(\lambda_1, \phi_1, \beta_1)$ and  $Y_2 \sim SMPtW(\lambda_2, \phi_2, \beta_2),$ then, the stress-strength parameter, say R, is defined as

$$R=\int_{-\infty}^{\infty} f_1(y;\lambda_1, \phi_1, \beta_1)F_2(y;\lambda_2, \phi_2, \beta_2) dy  .$$
For instance, suppose $Y_1$ represents the strength of a component, independent of $Y_2$, the component's stress. If $Y_1 < Y_2$, then either the component fails or the system that uses the component may malfunction. Hence, reliability can be expressed as the probability that $Y_1 >Y_2.$ That is,

$$R=P(Y_1>Y_2)=\int_{-\infty}^{\infty} f_1(y;\lambda_1, \phi_1, \beta_1 )F_2(y;\lambda_2, \phi_2, \beta_2) dy.$$
Now, using the CDF and the PDF of SMPtW distribution in equations \ref{SMPtWcdf equation} and  \ref{SMPtWpdf equation}, we obtained;

$$R=\int_0^{\infty}\bigg(\dfrac{e^{log_{\lambda_1}e^{-(y/\beta_1)^{\phi_1}}}log(\lambda_1)*\phi_1e^{-(y/\beta_1)^{\phi_1}}(y/\beta_1)^{\phi_1-1}}{\lambda_1-1}\bigg)\bigg(\dfrac{e^{log(\lambda_2)e^{-(y/\beta_2)^{\phi_2}}}-\lambda_2}{1-\lambda_2}\bigg)dy$$

\begin{equation}
\begin{split}
R=&\dfrac{log(\lambda_1)\phi_1}{(\lambda_1-1)(1-\lambda_2)}\bigg[\int_0^{\infty}e^{[log(\lambda_1)e^{-(y/\beta_1)^{\phi_1}}-(y/\beta_1)^{\phi_1}+log(\lambda_2)e^{-(y/\beta_2)^{\phi_2}}]}(y/\beta_1)^{\phi_1-1} dy - \lambda_2 \\
& \times \int_0^{\infty}e^{[log(\lambda_1)e^{-(y/\beta_1)^{\phi_1}}-(y/\beta_1)^{\phi_1}]}(y/\beta_1)^{\phi_1-1} dy\bigg]
\label{rel}
\end{split}
\end{equation}

The stress-strength reliability above is not in a closed form if the shape parameters $\phi_1 \ne \phi_2$ and $\beta_1\neq \beta_2$ and thus, must be evaluated numerically. However, if the shape parameters $\phi_1=\phi_2,$ and $\beta_1 = \beta_2 $, we have a closed form, and the Stress-Strength Reliability is given as

\begin{equation}
R= \dfrac{log(\lambda_1)}{(\lambda_1-1)(1-\lambda_2)}\bigg[\dfrac{\lambda_1 \lambda_2-1}{log(\lambda_1\lambda_2)}-\lambda_2\dfrac{\lambda_1-1}{log(\lambda_1)}\bigg]
\end{equation}

The reliability above is obtained by setting  $t=(y/\beta)^{\phi}$ in equation \ref{rel} so that $dy=\frac{\phi}{\beta}t^{\phi-1}dt$.

\subsection{Order Statistics}
Suppose $y_{(1)}, y_{(2)}, ..., y_{(n)}$ are random samples obtained from SMPtW distribution. Let $Y_{(j;n)}$ be the $j^{th}$ order statistics with SMPtW PDF. Hence, the probability distribution function of $y_{(j;n)}$ denoted as $f_{(j:n)} (y;\lambda, \phi, \beta )$ is given as:

\begin{equation}
f_{(j:n)} (y;\lambda, \phi, \beta ) = \frac{n!}{(j-1)! (n-j)!}[F_{\text{SMPtW}}(y;\lambda, \phi, \beta)]^{j-1}[1-F_{\text{SMPtW}}(y;\lambda, \phi, \beta)]^{n-j}f_{\text{SMPtW}}(y;\lambda, \phi, \beta) 
\label{SMPtWorderstatistics}
\end{equation}
Using the PDF and CDF of SMPtW distribution in equations \ref{SMPtWpdf equation} and \ref{SMPtWcdf equation} respectively, the $j^{th}$ ordered statistics of SMPtW distribution is derived following from equation \ref{SMPtWorderstatistics} as:

\begin{equation*}
\begin{split}
f_{(j:n)}(y;\lambda, \phi, \beta) =& \frac{y^{\phi-1}n!\phi (log \lambda) }{\beta(\lambda-1)(j-1)! (n-j)!} \left[\frac{e^{{\left( {\log \lambda } \right)} {  e^{-\left( \frac{y}{\beta }\right)^\phi} }}- \lambda} {1-\lambda}\right]^{j-1}  \\
\times
&\left[1 -\frac{e^{{\left( {\log \lambda } \right)} {  e^{-( \frac{y}{\beta })^\phi} }} +\lambda}{1-\lambda}\right]^{n-j} e^{[{\left( {\log \lambda } \right)} {  e^{-( \frac{y}{\beta })^\phi } } - \left(\frac{y}{\beta }\right)^\phi ]}
\end{split}
\end{equation*}

\begin{equation}
\begin{split}
f_{(j:n)} (y;\lambda, \phi, \beta)  = &\frac{n!\phi (log \lambda) y^{\phi -1 } \left[ 1-\lambda \right] ^{1-n}}{\beta^\phi (\lambda-1)(j-1)! (n-j)!} [e^{{\left( {\log \lambda } \right)} {  e^{-\left( \frac{y}{\beta }\right)^\phi} }}- \lambda]^{j-1} \\
&\times [1-e^{{\left( {\log \lambda } \right)} {  e^{-\left( \frac{y}{\beta }\right)^\phi} }}]^{n-j}e^{\left[{\left( {\log \lambda } \right)} {  e^{-\left( \frac{y}{\beta }\right)^\phi } } - \left(\frac{y}{\beta }\right)^\phi \right]}
\label{SMPtWorderstatistics_equation}
\end{split}
\end{equation}

We can then obtain the first and nth order statistics using equation \ref{SMPtWorderstatistics_equation}. Consequently, the first and nth order statistics are given in equations \ref{SMPtWorderstatistics1} and \ref{SMPtWorderstatistics2} respectively.

\begin{equation}
f_{(1:n)} (y;\lambda, \phi, \beta ) = \frac{n! \phi (log \lambda) y^{\phi -1 } \left( 1-\lambda \right) ^{1-n}}{\beta^\phi(\lambda-1)}
 \left[1-e^{{\left( {\log \lambda } \right)} {  e^{-\left( \frac{y}{\beta }\right)^\phi} }}\right]^{n-1} e^{\left[{\left( {\log \lambda } \right)} {  e^{-\left( \frac{y}{\beta }\right)^\phi } } - \left(\frac{y}{\beta }\right)^\phi \right]}
\label{SMPtWorderstatistics1}
\end{equation}

\begin{equation}
f_{(n:n)} (y;\lambda, \phi, \beta) = \frac{n! \phi (log \lambda) y^{\phi -1 } \left( 1-\lambda \right) ^{1-n}}{\beta^\phi(\lambda-1)}
 \left[e^{{\left( {\log \lambda } \right)} {  e^{-\left( \frac{y}{\beta }\right)^\phi} }} - \lambda \right]^{n-1} e^{\left[{\left( {\log \lambda } \right)} {  e^{-\left( \frac{y}{\beta }\right)^\phi } } - \left(\frac{y}{\beta }\right)^\phi \right]}
\label{SMPtWorderstatistics2}
\end{equation}

\subsection{Rényi Entropy}
In this section, we derived the Rényi entropy denoted as $E_r (y; \lambda, \phi, \beta )$ given by

\begin{equation}
    E_r (y; \lambda, \phi, \beta ) =\frac{1}{1-h} log\left (\int_{0}^{\infty} f_{\text{SMPtW}}(y;\lambda, \phi, \beta )^h dy \right); \quad \quad h>0, h \ne 1
    \label{SMPtWrenyi}
\end{equation}

\begin{theorem}
The Rényi entropy of a random variable which follows the SMPtW distribution is given as:

\begin{equation*}
  E_r (y; \lambda, \phi, \beta)  = \frac{1}{1-h} \text{log} \left[\left( \frac{ (\log \lambda) \phi }{\beta (\lambda - 1)} \right)^h 
\frac{\beta}{\phi} \, \Gamma\Big(\frac{\phi h - h + 1}{\phi}\Big) 
\sum_{j=0}^{\infty} \frac{(h \log \lambda)^j}{j!} (h+j)^{- \frac{\phi h - h + 1}{\phi}} \right ]
\end{equation*}
where $\beta > 0, \phi >0$ and $h+j>0$
\end{theorem}
\begin{proof}

Rényi entropy is given by

\begin{equation}
    E_r (y; \lambda, \phi, \beta ) =\frac{1}{1-h} log\left (\int_{0}^{\infty} f_{\text{SMPtW}}(y;\lambda, \phi, \beta )^h dy \right); \quad \quad h>0, h \ne 1
   \end{equation}
Now, using the PDF given in equation \ref{SMPtWpdf equation}, we have;
\begin{equation*}
      E_r (y; \lambda, \phi, \beta) =\frac{1}{1-h} log\left (\int_{0}^{\infty} \left[ \frac{{e^{{\left( {\log \lambda } \right)} {  e^{-\left( \frac{y}{\beta }\right)^\phi} }}} }{{\lambda  - 1}} (\log \lambda) \frac{\phi }{\beta }\left( {\frac{y}{\beta }} \right)^{\phi  - 1} e^{ - \left( {\frac{y}{\beta }} \right)^\phi}\right]^h dy \right)
\end{equation*}

\begin{equation*}
      E_r (y; \lambda, \phi, \beta ) =\frac{ 1 }{1-h} log\left ( {\left[\frac{\phi log \lambda}{ \beta (\lambda - 1)}\right]}^h \int_{0}^{\infty} \left(\frac{y}{\beta}\right) ^{ h(\phi -1)} e^{-h \left(\frac{y}{\beta}\right)^\phi} e^{{h\left(\log \lambda \right)}e^{-\left(\frac{y}{\beta}\right)^ \phi}} dy \right)
\end{equation*}
Using the Taylor series expansion 
\begin{equation*}
    e^{{h\left(\log \lambda \right)}e^{-\left(\frac{y}{\beta}\right)^ \phi}} = \sum_{j=0}^\infty \dfrac{(hlog \lambda)^j}{j !} e^{-j \left(\frac{y}{\beta}\right)^\phi}
\end{equation*}
we have:
\begin{equation*}
        E_r (y; \lambda, \phi, \beta ) =\frac{ 1 }{1-h} log\left ( {\left[\frac{\phi log \lambda}{ \beta (\lambda - 1)}\right]}^h \int_{0}^{\infty} \left(\frac{y}{\beta}\right) ^{ h(\phi -1)} e^{-h \left(\frac{y}{\beta}\right)^\phi} \sum_{j=0}^\infty \dfrac{(hlog \lambda)^j}{j !} e^{-j \left(\frac{y}{\beta}\right)^\phi} dy \right)
\end{equation*}
\begin{equation*}
        E_r (y; \lambda, \phi, \beta ) =\frac{ 1 }{1-h} log\left ( {\left[\frac{\phi log \lambda}{ \beta (\lambda - 1)}\right]}^h \sum_{j=0}^\infty \dfrac{(hlog \lambda)^j}{j !}\int_{0}^{\infty} \left(\frac{y}{\beta}\right) ^{ h(\phi -1)} e^{-h \left(\frac{y}{\beta}\right)^\phi}  e^{-j \left(\frac{y}{\beta}\right)^\phi} dy \right)
\end{equation*}
\begin{equation}
        E_r (y; \lambda, \phi, \beta ) =\frac{ 1 }{1-h} log\left ( {\left[\frac{\phi log \lambda}{ \beta (\lambda - 1)}\right]}^h \sum_{j=0}^\infty \dfrac{(hlog \lambda)^j}{j !}\int_{0}^{\infty} \left(\frac{y}{\beta}\right) ^{ h(\phi -1)} e^{-(h+j) \left(\frac{y}{\beta}\right)^\phi} dy \right)
        \label{renyi sub}
\end{equation}
To simplify, let $ z = \left( \frac{y}{\beta} \right)^\phi$,

\noindent Therefore, $ y = \beta z^{(1/\phi)}$ \quad $dy = \frac{\beta}{\phi} z^{ \frac{1}{\phi} - 1}dz$ and $ \left( \frac{y}{\beta} \right)^{h(\phi-1)} = z^{\frac{h(\phi - 1)}{\phi}}$

\noindent Hence, 
$$ \int_{0}^{\infty} \left(\frac{y}{\beta}\right) ^{ h(\phi -1)} e^{-(h+j) \left(\frac{y}{\beta}\right)^\phi} dy = \frac{\beta}{\phi} \int_0^\infty z^{\frac{h(\phi - 1)}{\phi} +\frac{1}{\phi}-1} e^{-(h+j)z} dz$$

$$ \int_{0}^{\infty} \left(\frac{y}{\beta}\right) ^{ h(\phi -1)} e^{-(h+j) \left(\frac{y}{\beta}\right)^\phi} dy = \frac{\beta}{\phi} \int_0^\infty z^{\frac{h\phi - h +1}{\phi} -1} e^{-(h+j)z} dz$$

\noindent Using the general gamma function formula, we have:
\begin{equation}
\int_{0}^{\infty} \left(\frac{y}{\beta}\right) ^{ h(\phi -1)} e^{-(h+j) \left(\frac{y}{\beta}\right)^\phi} dy = \frac{\beta}{\phi} \Gamma \left(\frac{h\phi - h +1}{\phi}\right) (h+j)^{- \left(\frac{h\phi - h +1}{\phi}\right)}
\label{gamma}
\end{equation}

Substituting equation \ref{gamma} in equation \ref{renyi sub}, we obtained the Rényi entropy for SMPtW distribution as 
\begin{equation}
  E_r (y; \lambda, \phi, \beta)  = \frac{1}{1-h} \text{log} \left[\left( \frac{ (\log \lambda) \phi }{\beta (\lambda - 1)} \right)^h 
\frac{\beta}{\phi} \, \Gamma\Big(\frac{\phi h - h + 1}{\phi}\Big) 
\sum_{j=0}^{\infty} \frac{(h \log \lambda)^j}{j!} (h+j)^{- \frac{\phi h - h + 1}{\phi}} \right ]
\end{equation}
\end{proof}
\section{Maximum Likelihood Estimation}
The likelihood function for the SMPtW distribution obatined as:
\begin{equation*}
    L(\theta|\boldsymbol{y})= \prod_{i=1}^nf(y_i|\lambda, \phi, \beta)
\end{equation*}

\begin{equation*}
    L= \prod_{i=1}^n \dfrac{e^{[(log \lambda)e^{-(y_i/\beta)^{\phi}}-(y_i/\beta)^{\phi}]}(log \lambda) (\phi/\beta) (y_i/\beta)^{\phi -1 }}{\lambda-1}
\end{equation*}

Hence, $\log$ likelihood function is given by;
\begin{equation}
\begin{split}
    Q=log L=&log \lambda \sum_{i=1}^ne^{-(y_i/\beta)^{\phi}} + nlog(log\lambda)+nlog(\phi/\beta)+(\phi - 1)\sum_{i=1^n}( log (y_i/ \beta)) \\
    &- \sum_{i=1}^n(y_i/\beta)^{\phi}-nlog(\lambda-1)
    \label{loglik}
    \end{split}
\end{equation}

The MLEs of $\lambda$, $\phi$ and $\beta$ are obtained by partially differentiating $Q$ in equation \ref{loglik} with respect to $\lambda$, $\phi$ and $\beta$ and equating to zero as follows:

\begin{equation}
    \frac{\partial Q}{\partial \lambda}=\frac{1}{\lambda}\sum_{i=1}^ne^{-(y_i/\beta)^{\phi}}+\frac{n}{\lambda log \lambda}+\frac{n}{\lambda-1}=0
    \label{partiallambda}
\end{equation}

\begin{equation}
    \frac{\partial Q}{\partial \phi}=-(log \lambda) \sum_{i=1}^n (y_i/\beta)^{\phi} log (y_i/\beta)  e^{(y_i/\beta)^{\phi}}+\frac{n\beta}{\phi}+\sum_{i=1}^nlog (y_i/\beta) - \sum_{i=1}^n(y_i/\beta)^{\phi}log (y_i/\beta) 
       \label{partialphi}
\end{equation}

\begin{equation}
    \frac{\partial Q}{\partial \beta}=-\frac{\phi log \lambda }{\beta}\sum_{i=1}^n(y_i/\beta)^{\phi}e^{-(y_i/\beta)^{\phi}}-\frac{1}{\beta}+\frac{\phi}{\beta}\sum_{i=1}^n(y_i/\beta)^{\phi}+\frac{n(\phi-1)}{\beta}
    \label{partialbeta}
\end{equation}

Note that there are no closed form for the above partial derivatives. The MLEs $(\hat{\lambda}, \hat{\phi}, \hat{\beta})$ are obtained by simultaneously solving equations \ref{partiallambda}, \ref{partialphi} and \ref{partialbeta} using numerical methods such as Secant,  Regula-Falsi or Newton-Raphson.

\section{Simulation study}

\subsection{Random number generation} \label{random}
In generating random numbers for the SMPtW distribution, the inverse transform approach is adopted in this study. This method has been widely used by several authors (e.g. \citet{adesegun2023transmuted}, \citet{menberu2024transformed}) in generating random numbers in distribution theory. The algorithm for generating random numbers from the SMPtW distribution using the inverse transform approach is given as:

\begin{algorithm}
\caption{Random number generation from SMPtW distribution}\label{algo1}
(i) Generate a random variable say $U$ such that $U \sim U( 0, 1)$

(ii) Obtain the inverse of the CDF of the SMtW distribution

(iii) Apply  $X = F(y, \lambda, \phi, \beta )^{-1} (U)$ 

(iv) Then $X \sim F(y; \lambda, \phi, \beta )$

(v) Repeat until the required number of samples are obtained
\end{algorithm}
\subsection{Simulation result}
In this section, we examined the performance of Maximum likelihood methods in estimating the parameters of the SMPtW distribution. Random samples of size 50, 100, 250, 500 and $1,000$ were simulated from SMPtW distribution using the inverse transform method discussed in subsection \ref{random} over $1,000$ iterations each. We considered the parameter combinations; $ (\lambda = 0.3, \phi =0.5, \beta = 0.2),  (\lambda = 1, \phi =1, \beta = 1.5),  (\lambda = 1, \phi =0.5, \beta = 1), (\lambda = 3, \phi = 4, \beta =0.5) $, and $(\lambda = 1.5, \phi = 2.5, \beta =0.5) $. The shape parameter ($\phi$) values were chosen to reflect heavy, moderate and light tail behaviour of the distribution. Also, the scale ($\beta$) was selected to reflect the compressed, natural and stretched nature of the distribution, while the transformation parameter ($\lambda$) was chosen to consider exact Weibull and sub-cases of the distribution.  We evaluated the MLE estimates based on the bias and mean square error. 
\begin{table}[h]
\centering
\caption{Parameter estimates for the simulated data}\label{simulation}%
\begin{tabular}{@{}llllllllllllll@{}}
\toprule
S/N & \multicolumn{3}{c}{Parameters}  & Size & \multicolumn{3}{c}{Estimates}  & \multicolumn{3}{c}{Bias}  & \multicolumn{3}{c}{MSE} \\
\midrule
& $\lambda$ & $\phi$ & $\beta$  & n &  $\hat{\lambda}$ & $\hat{\phi}$ & $\hat{\beta}$  &  $\hat{\lambda}$ & $\hat{\phi}$ & $\hat{\beta}$  & $\hat{\lambda}$ & $\hat{\phi}$ & $\hat{\beta}$\\
\midrule
1& 0.3 & 0.5 & 0.2 & 50 & 18.5276 & 0.5445 & 0.4969 & 18.2276 & 0.0445 & 0.2969 
     & 43049.44 & 0.0176 & 0.8757  \\
&&& &100 & 22.63 & 0.5307 & 0.4231 
     & 22.3335 & 0.0307 & 0.2231 
     & 58641.94 & 0.0122 & 0.7180 \\
&&& &250 & 34.5524 & 0.5168 & 0.3232 
     & 34.2524 & 0.0168 & 0.1232 
     & 593216.03 & 0.0074 & 0.3650 \\
&& & &500 & 9.5389  & 0.5055 & 0.2641 
     & 9.2389  & 0.0055 & 0.0641 
     & 28536.28 & 0.0047 & 0.2361 \\
&& & &1000 & 0.4248  & 0.5039 & 0.2222 
     & 0.1248  & 0.0039 & 0.0222 
     & 0.1741 & 0.0028 & 0.0108 \\

\botrule
2& 1.0  & 1.0 & 1.5 & 50 & 13.4553 & 0.9963 & 1.7100  & 12.4553 & -0.0037 & 0.2100 
     & 7653.99 & 0.0404 & 1.2903 \\
&&& &100 & 17.6098 & 0.9856 & 1.6297 
     & 16.6098 & -0.0144 & 0.1297 
     & 24682.96 & 0.0314 & 1.0377 \\
&&& &200 & 52.9653 & 0.9824 & 1.6081 
     & 51.9653 & -0.0176 & 0.1081 
     & 382859.75 & 0.0193 & 0.9306 \\
&& & &500 & 30.8381 & 0.9838 & 1.5342 
     & 29.8381 & -0.0162 & 0.0342 
     & 295017.80 & 0.0129 & 0.5424 \\
&& & &1000 & 1.1898  & 0.9931 & 1.5021 
     & 0.1898  & -0.0069 & 0.0021 
     & 0.6921 & 0.0066 & 0.0950 \\
\botrule
3& 1.0 & 0.5 & 1.0 & 50  & 13.2051 & 0.4860 & 1.7427   & 12.2051 & -0.0140 & 0.7427 
     & 7662.86 & 0.0128 & 10.7227 \\
&&& &100 & 17.2258 & 0.4814 & 1.5143 
     & 16.2258 & -0.0186 & 0.5143 
     & 24675.78 & 0.0099 & 10.9227 \\
&&& &250 & 52.8851 & 0.4876 & 1.5271 
     & 51.8851 & -0.0124 & 0.5271 
     & 382847.61 & 0.0053 & 16.0094 \\
&& & &500 & 30.8320 & 0.4913 & 1.2821 
     & 29.8320 & -0.0087 & 0.2821 
     & 295018.26 & 0.0033 & 9.2383 \\
&& & &1000 & 1.1874  & 0.4963 & 1.0434 
     & 0.1874  & -0.0037 & 0.0434 
     & 0.6909 & 0.0017 & 0.1671 \\
\botrule
4& 3.0 & 4.0 & 2.0 & 50 & 23.4346 & 3.7528 & 1.8978 & 20.4346 & -0.2472 & -0.1022 
     & 65374.64 & 0.6560 & 0.1290 \\
&&& &100 & 17.3497 & 3.7796 & 1.9197 
     & 14.3497 & -0.2204 & -0.0803 
     & 9733.76 & 0.4960 & 0.0987 \\
&&& &250 & 40.2449 & 3.8301 & 1.9591 
     & 37.2449 & -0.1699 & -0.0409 
     & 124896.39 & 0.3172 & 0.0693 \\
&& & & 500& 46.6442 & 3.9110 & 1.9829 
     & 43.6442 & -0.0890 & -0.0171 
     & 289260.07 & 0.1605 & 0.0427 \\
&& & &1000 & 3.4999  & 3.9731 & 1.9926 
     & 0.4999  & -0.0269 & -0.0074 
     & 5.4673 & 0.0592 & 0.0120 \\
\botrule
5& 1.5 & 2.5 & 0.5 & 50 & 16.3178 & 2.4380 & 0.4932   & 14.8178 & -0.0620 & -0.0068 
     & 24447.89 & 0.2370 & 0.0157 \\
&&& &100 & 69.2060 & 2.4172 & 0.4917 
     & 67.7060 & -0.0828 & -0.0083 
     & 581282.37 & 0.1915 & 0.0144 \\
&&& &250 & 50.9219 & 2.4271 & 0.4945 
     & 49.4219 & -0.0729 & -0.0055 
     & 1075946.55 & 0.1176 & 0.0087 \\
&& & &500 & 16.9683 & 2.4516 & 0.4948 
     & 15.4683 & -0.0484 & -0.0052 
     & 52728.26 & 0.0716 & 0.0053 \\
&& & &1000 & 1.7479  & 2.4808 & 0.4972 
     & 0.2479  & -0.0192 & -0.0028 
     & 1.3381 & 0.0335 & 0.0019 \\
\botrule
\end{tabular}
\end{table}

The results of the simulation are presented in Table \ref{simulation}. As the sample size increases, the estimates approach the true parameter values simulated. The Bias and MSE also reduce as the sample sizes increase for $\phi$ and $\beta$, but some inconsistencies were noticed for $\lambda$. See the discussion section for details.  

\section{Application with Health data.}
In this section, we compared the proposed SMPtW distribution with the base Weibull distribution and some similar extensions of  Weibull distribution previously studied in the literature. These distributions include Exponentiated Weibull \citep{pal2006exponentiated}, transmuted Weibull \citep{aryal2011transmuted} and Sine alpha power Weibull \citep{alghamdi2025sine} distributions. The PDFs of the distributions compared with that of the SMPtW distribution are given in Table \ref{compared distributions}. 

The data used in showing the supremacy of SMPtW distribution over other similar distributions consists of seventy-six Kelvar 373/epoxy fatigue fracture (\citet{zhao2023novel}, \citet{owoloko2015performance} and \cite{rasool2025innovative}). The performance of the fitted distributions were assessed based on Akaike information criteria (AIC), Bayesian information criteria (BIC), corrected Akaike information criteria (AICc), Hannan–Quinn information criterion (HQIC), Kolmogorov-Smirnov (K-S) statistics and corresponding p-values.

\begin{table}[h]
\caption{PDF of compared distributions}\label{compared distributions}%
\begin{tabular}{@{}llll@{}}
\toprule
SN & Distributions & Probability distribution function \\
\midrule
1 & Weibull    & $f(y) = 
 \frac{\phi }{\beta }\left( {\frac{y}{\beta }} \right)^{\phi  - 1} e^{ - \left( {\frac{y}{\beta}} \right)^\phi  } \quad y\geq 0.$    \\
2 & Exponentiated Weibull    & $
f(y) = \beta \phi \lambda y^{\phi - 1}
e^{-\lambda y^{\phi}}
\left(1 - e^{-\lambda y^{\phi}}\right)^{\beta - 1},
\quad y \ge 0. $ \\
3 & Transmuted Weibull   &  $
f(y) =
\phi \lambda y^{\phi - 1}
e^{-\lambda y^{\phi}}
\left[
1 + \beta - 2\beta \left(1 - e^{-\lambda y^{\phi}}\right)
\right],
\quad y \ge 0.
$   \\
4 &  Sine alpha power Weibull    & $
f(y)
=
\frac{\pi}{2}
\,\lambda \beta
\left(\frac{\log(\phi)}{\phi - 1}\right)
y^{\beta-1}
e^{-\lambda y^{\beta}}
\phi^{\,1 - e^{-\lambda y^{\beta}}}
\cos\left[
\frac{\pi}{2}
\left(
\frac{1 - \phi^{\,1 - e^{-\lambda y^{\beta}}}}{1 - \phi}
\right)
\right],
\quad y>0.$  \\
5 & SMP transformed Weibull  &  $
f(y) = \,\left\{ \begin{array}{l}
  \frac{{e^{{\left( {\log \lambda } \right)} {  e^{-\left( \frac{y}{\beta }\right)^\phi} }}} }{{\lambda  - 1}} (\log \lambda) \frac{\phi }{\beta }\left( {\frac{y}{\beta }} \right)^{\phi  - 1} e^{ - \left( {\frac{y}{\beta }} \right)^\phi}\,\,\,\,\,\,\,\,\,\,if\,\lambda  > 0,\,\,\lambda  \ne 1, \\ 
\frac{\phi }{\beta }\left( {\frac{y}{\beta }} \right)^{\phi  - 1} e^{ - \left( {\frac{y}{\beta }} \right)^\phi  }\,\,\,\,\,\,\,\,\,\,\,\,\,\,\,\,\,\,\,\,\,\,\,\,\,\,\,\,\,\,\,\,\,\,\,\,\,\,\,\,\,\,\,\,if\,\lambda  = 1\,\, \\ 
 \end{array} \right.
$  \\
\botrule
\end{tabular}
\end{table}
\begin{table}[h]
\caption{MLE Estimates for Fracture data with their standard errors in brackets}\label{fracture}%
\begin{tabular}{@{}llllll@{}}
\toprule
SN & Distributions & \multicolumn{3}{c}{Parameters}  \\
\midrule
&     &  $\hat{\lambda}$       & $\hat{\phi} $ & $\hat{ \beta }$  \\
\midrule
1 & Two parameter Weibull    &    & 1.3256 (0.1138)  & 2.1328 (0.1944) \\
2 & Exponentiated Weibull    & 1.4427 (0.6436) &  1.1013 (0.2628) & 1.6409 (0.5871) \\
3 & Transmuted Weibull   &  -0.7955 (0.2130) & 1.0509 (0.1262)  & 1.4419 (0.2120) \\
4 &  Sine alpha power Weibull  & 0.0546 (0.1537) & 1.4122 (0.1290) & 6.4225 (3.9454) \\
5 & SMP transformed Weibull  &  0.0102 (0.0184) & 0.7849 (0.1465)  & 0.7061 (0.2919)  \\
\botrule
\end{tabular}
\label{comparison 2}
\end{table}

\begin{table}[h]
\caption{Evaluation metrics result for data on fracture}\label{Fracture comparison}%
\begin{tabular}{@{}lllllllll@{}}
\toprule
SN & Distributions & \multicolumn{6}{c}{Evaluations metrics}  \\
&     &  AIC     & BIC & AICs &  HQIC & KS & P-Value\\
\midrule
1 & Two parameter Weibull    &    249.0494 & 253.7108 & 249.2138 & 250.9123 & 0.1099 & 0.2953\\
2 & Exponentiated Weibull    & 250.3272 & 257.3194 & 250.6606 & 253.1216 & 0.0988 &  0.4217\\
3 &  Sine alpha power Weibull  & 249.0571 & 256.0493 &  249.3904 & 251.8515 & 0.0919 &  0.5133 \\
4 & Transmuted Weibull   &   248.8600 & 255.8522 & 249.1933 &251.6544 & 0.0988& 0.4213\\
5 & SMP transformed Weibull  & 247.3419 & 254.3341 & 247.6752 & 250.1363 & 0.0881 & 0.5668 \\
\botrule
\end{tabular}
\label{MLE 2}
\end{table}
Tables \ref{comparison 2} and \ref{MLE 2} show the MLE and comparison metrics. The result showed that the three-parameter SMP transformed Weibull distribution performed best among the competing models based on the evaluation metrics.

\begin{figure}[H]
\centering
\begin{minipage}{0.45\textwidth}
    \centering
    \includegraphics[width=\linewidth]{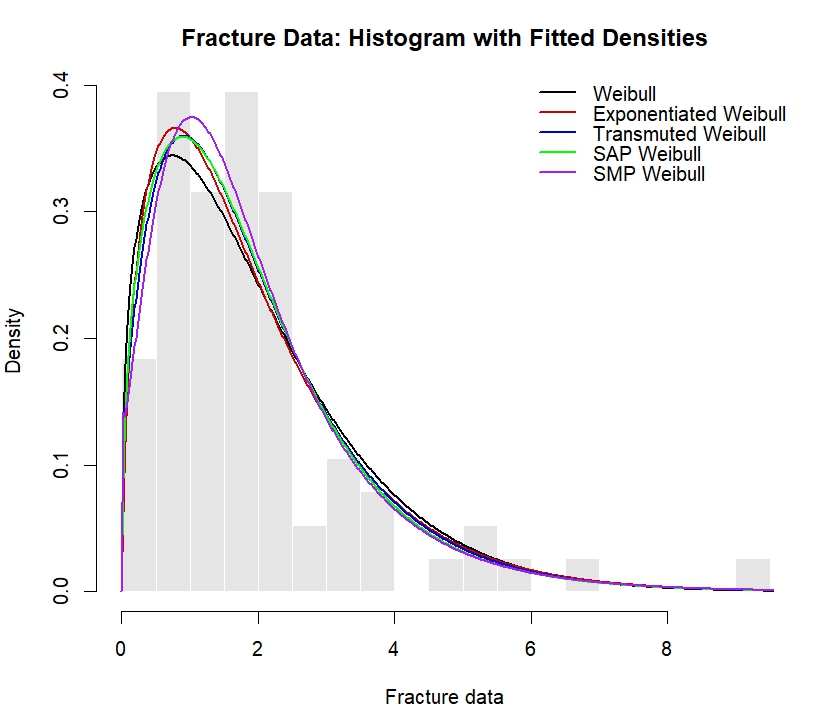}
    \caption{Histogram with Fitted Densities}
    \label{fig5}
\end{minipage}\hfill
\begin{minipage}{0.45\textwidth}
    \centering
    \includegraphics[width=\linewidth]{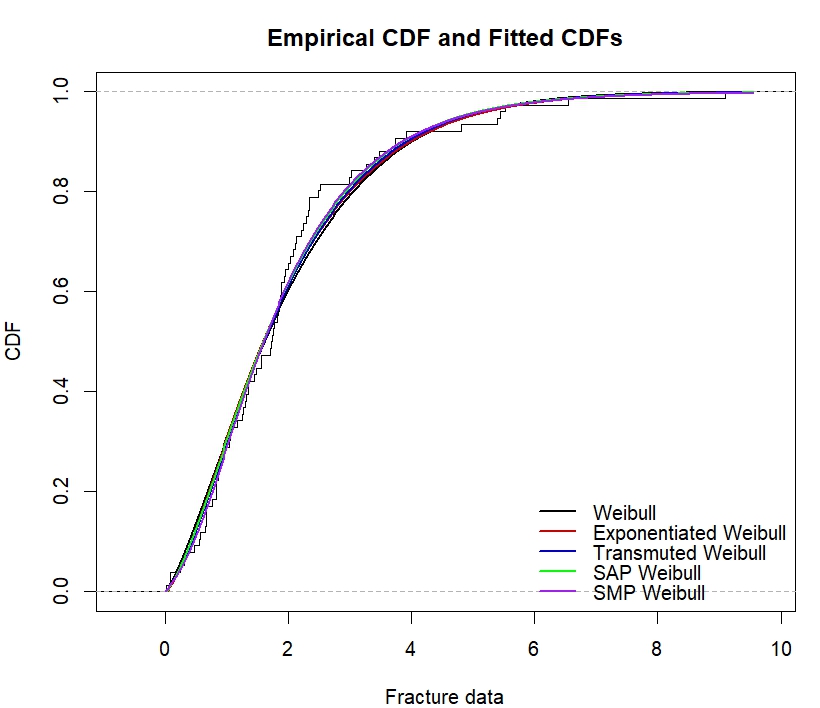}
    \caption{Empirical CDF and Fitted CDFs}
    \label{fig6}
\end{minipage}
\end{figure}

The histogram with the fitted densities of the distributions is shown in Figures \ref{fig5} while the empirical CDF versus the fitted CDF plot is shown in Figure \ref{fig6}. It was observed that the SMPtW distribution fits the data best among the competing models.

\begin{figure}[H]
\centering
   \includegraphics[width=0.8\linewidth]{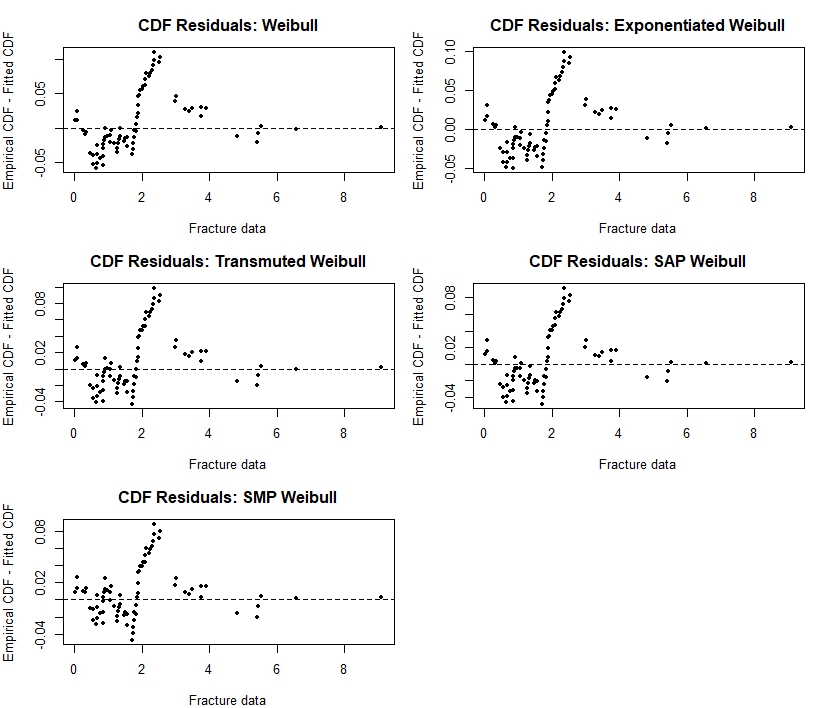}
    \caption{Cumulative Distribution Function Residuals of the fitted Distributions}
    \label{fig7}
\end{figure}

Figure \ref{fig7} shows the  CDF residuals for the fitted distributions. SMPtW CDF residuals tend to be normally distributed compared to other fitted distributions, corroborating the fact that SMPtW fits the data best among the competing models. 

\section{Discussion}

The simulation result in Table \ref{simulation} highlight a clear distinction in the finite-sample behaviour of the maximum likelihood estimators for $(\lambda,\phi,\beta)$ under the SMPtW model. The baseline Weibull parameters $\phi$ and $\beta$ are consistently well-estimated across all scenarios, exhibiting small bias and rapidly decreasing mean squared error (MSE) as the sample size increases, with estimates concentrating around their true values even in moderate samples; this reflects strong identifiability of the underlying Weibull structure. In contrast, the transformation parameter $\lambda$ shows substantial bias and pronounced variability, particularly in small and moderate samples, with MSE values that remain large relative to those of $\phi$ and $\beta$; although performance improves with increasing sample size, convergence is noticeably slower and less stable. Future studies could consider robust estimation techniques, which would be valuable in addressing potential numerical challenges associated with the distribution.

Table \ref{MLE 2} provides a comprehensive comparison of the fitted model to the fracture dataset using AIC, BIC, AICc, HQIC, K-S and associated p-values as evaluation metrics. The classical Weibull model performs poorly across all criteria, with an AIC of 294.0494, which is substantially larger than all competing models. This clearly indicates that the baseline Weibull model is inadequate for the fracture data. Among the extended models, the exponentiated Weibull model records an AIC of 250.3272, which is higher than several competitors, suggesting that its particular form of extension does not align well with the structure of the data. The transmuted Weibull and sine alpha power Weibull models perform better, with AIC values of 248.8600 and 249.0571, respectively. These results indicate that introducing additional shape flexibility improves model fit, although the extent of improvement varies across transformation mechanisms.

The proposed SMP transformed Weibull model achieves the lowest AIC value of 247.3419. Importantly, the SMP transformed Weibull model is consistently ranked first across all four criteria (AIC, BIC, AICc, HQIC), which strengthens the reliability of this conclusion. Not only that, it has the lowest K-S values (0.0881) with the highest P-value (0.5668), which shows that it has the smallest and statistically significant deviation between the empirical and the fitted CDFs as shown in Figures \ref{fig5} and \ref{fig6}.

The superior performance of the SMP transformed Weibull model indicates that its transformation mechanism is particularly well-suited to capture these complexities in the data. Unlike standard extensions that adjust the baseline distribution in a relatively rigid manner, the SMP transformation introduces a more flexible deformation of the distribution, allowing it to better adapt to variations in both the central body and the tails of the data. Another important observation is that all extended Weibull-type models outperform the classical Weibull distribution. This reinforces the idea that additional flexibility is not optional but necessary when modelling real-world lifetime data. However, the extent of improvement varies, and the SMP transformed Weibull model provides the most consistent and favourable performance across all evaluation criteria.

What distinguishes the proposed model is not simply the addition of parameters, but the manner in which the transformation reshapes the distribution. This allows for greater adaptability in modelling complex phenomena where traditional assumptions about non-monotonic hazard rates and simple tail structures are often violated.

\section{Conclusion}
Prior to introducing the novel distribution of interest, we gave an update on some methods of extending distributions recently proposed in the literature that are yet to gain much attention. We also gave a brief review of transformed distributions through SMP method. This is to give the gaps at a glance on base distributions that are yet to be explored via this method. While the base two-parameter Weibull distribution remains a foundational tool in statistical modelling, its structural limitations necessitate the development of more flexible alternatives. We introduced and studied a novel three-parameter Weibull distribution obtained via SMP technique, named SMPtW distribution. The properties of this distribution were derived and studied extensively. Through the inverse transform random number generation approach, we evaluated the consistency of MLE estimates for the parameters of SMPtW distribution via simulation. Bias of estimates and MSE were used as evaluation metrics. Furthermore, we showed the supremacy of the SMPtW distribution over other similar distributions in modelling health data using AIC, BIC, AICc, HQIC, K-S and corresponding p-values. It is our expectation that modellers can make use of this distribution in modelling heavily tailed health data. Also, distribution theorists and practitioners can explore new methods identified here in extending distributions. This work contributes to the
growing literature on distributional extensions and provides a room for further methodological developments. In future, the development of more efficient estimation procedures, including Bayesian approaches and robust estimation techniques, would be valuable in addressing potential numerical challenges associated with the distribution.

\section*{Acknowledgements}

The authors thank anonymous colleagues for their constructive review. 

\section*{Declarations}

\subsection*{Authors Contributions}All the Authors contributed equally.

\subsection*{Funding} No funding was received for this research. 

\subsection*{Conflict of interest} The authors declared no conflict of interest
 
\subsection*{Ethics approval and consent to participate} Not required

\subsection*{Data availability} The data used is secondary data and cited accordingly within the paper.

\begin{appendices}\label{secA1}

\end{appendices}

\bibliography{sn-bibliography}

\end{document}